\documentclass[envcountsame]{llncs}

\usepackage{amssymb,comment}
\usepackage{amsmath}
\usepackage{xspace}
\usepackage{tikz}
\usepackage{pgflibrarysnakes}
\usetikzlibrary{snakes}
\usepackage{amstext}
\usepackage{calc}
\usepackage{amsopn}
\usepackage[noend]{algorithmic}
\usepackage[boxed]{algorithm}
\usepackage{eucal}
\usepackage{latexsym}
\usepackage{tabularx}

\newcommand{\compass}{\ensuremath{\textrm{NP} \subseteq \textrm{coNP}/\textrm{poly}}\xspace}

\newcommand{\cc}{{\rm{cc}}}

\floatname{algorithm}{Algorithm}

\newcommand{\defproblemu}[4]{
  \vspace{1mm}
\noindent\fbox{
  \begin{minipage}{\textwidth-0.5cm}
  \begin{tabular*}{\textwidth}{@{\extracolsep{\fill}}lr} #1 \\ \end{tabular*}
  {\bf{Input:}} #2  \\
  {\bf{Parameter:}} #3\\
  {\bf{Question:}} #4
  \end{minipage}
  }
  \vspace{1mm}
}

\newcommand{\defproblemugoal}[4]{
  \vspace{1mm}
\noindent\fbox{
  \begin{minipage}{\textwidth-0.5cm} \begin{tabular*}{\textwidth}{@{\extracolsep{\fill}}lr} #1  \\ \end{tabular*}
  {\bf{Input:}} #2  \\
  {\bf{Parameter:}} #3\\
  {\bf{Goal:}} #4
  \end{minipage}
  }
  \vspace{1mm}
}

\newcommand{\cbipst}{{$\oplus$-\sc{Bipartite Steiner Tree}}\xspace}
\newcommand{\wbipst}{{\sc{Weighted Bipartite Steiner Tree}}\xspace}
\newcommand{\bipst}{{\sc{Bipartite Steiner Tree}}\xspace}
\newcommand{\wcvcshort}{{\sc{WCVC}}\xspace}
\newcommand{\ccvcshort}{{\sc{$\oplus$-CVC}}\xspace}

\newcommand{\compwcvc}{{\sc{Compression Weighted Connected Vertex Cover}}\xspace}
\newcommand{\compwcvcshort}{{\sc{Comp-WCVC}}\xspace}
\newcommand{\cvc}{{\sc{Connected Vertex Cover}}\xspace}
\newcommand{\ccvc}{{$\oplus$-\sc{Connected Vertex Cover}}\xspace}
\newcommand{\wcvc}{{\sc{Weighted Connected Vertex Cover}}\xspace}

\newcommand{\weight}{\omega}

\begin{document}

  \date{}

\author{
  Marek Cygan
}

\institute{Institute of Informatics, University of Warsaw, Poland \\
      \email{cygan@mimuw.edu.pl}
}

  \title{Deterministic parameterized connected vertex cover 
  }

  \maketitle

\begin{abstract}
In the \cvc problem we are given an undirected graph $G$
together with an integer $k$ and we are to find
a subset of vertices $X$ of size at most $k$,
such that $X$ contains at least one end-point of each edge
and moreover $X$ induces a connected subgraph.
For this problem we present a deterministic algorithm running in $O(2^kn^{O(1)})$
time and polynomial space, improving over previously best $O(2.4882^kn^{O(1)})$
deterministic algorithm and $O(2^kn^{O(1)})$ randomized algorithm.
Furthermore, when usage of exponential space is allowed,
we present an $O(2^kk(n+m))$ time algorithm
that solves a more general variant with arbitrary real weights.

Finally, we show that in $O(2^kk(n+m))$ time and $O(2^kk)$ space
one can count the number of connected vertex covers of size at most $k$,
which can not be improved to $O((2-\varepsilon)^kn^{O(1)})$ for any $\varepsilon>0$
under the Strong Exponential Time Hypothesis, as shown by Cygan et al.~[CCC'12].
\end{abstract}

\section{Introduction}

In the classical vertex cover problem we are asked whether
there exists a set of at most $k$ vertices, containing
at least one end-point of each edge.
As a basic problem in the graph theory {\sc Vertex Cover}
is extensively studied, together with its natural variants.
One of the generalizations of {\sc Vertex Cover}
is the \cvc problem, where a vertex cover
is called a {\em connected vertex cover} if it induces
a connected subgraph.

\defproblemu{\cvc}{An undirected graph $G=(V,E)$ and an integer $k$.}{$k$}
{ Does there exist a connected vertex cover of $G$ of cardinality at most $k$?}

As \cvc is NP-complete we can not hope for polynomial time solutions,
however it is possible to efficiently solve the problem for small
values of $k$.
Obviously, for any fixed $k$, we can solve the problem
in polynomial time, by trying all $n^k$ possible subsets of vertices.
In the parameterized complexity setting we are interested
in finding algorithms of $f(k)n^{O(1)}$ running time,
for some computable function $f$, that is polynomial for each fixed value of $k$,
but where the degree of the polynomial is independent of $k$.

A few fixed-parameter algorithms were designed for the \cvc problem
during the last years.
The fastest deterministic algorithm is due to Binkele-Raible~\cite{raible-cvc}
running in $O^*(2.4882^k)$ time,
while the fastest (randomized) algorithm is due to Cygan et al.~\cite{focs}
running in $O^*(2^k)$ time, where by $O^*$ we denote the standard big $O$ notation,
with polynomial factors omitted.
In Table~\ref{tab1} we summarize the history of parameterized algorithms
for \cvc.

\begin{table*}[h]
\centering
\begin{tabular}{|c|c|}
\hline
$O^*(6^k)$ & Guo et al.~\cite{guo-cvc} \\
\hline
$O^*(3.2361^k)$ & M{\"o}lle et al.~\cite{molle-cvc1} \\
\hline
$O^*(2.9316^k)$ & Fernau et al.~\cite{fernau-cvc} \\
\hline
$O^*(2.7606^k)$ & M{\"o}lle et al.~\cite{molle-cvc2}\\
\hline
$O^*(2.4882^k)$ &Binkele-Raible~\cite{raible-cvc}\\
\hline
$O^*(2^k)$(randomized) &Cygan et al.~\cite{focs}\\
\hline
$O^*(2^k)$ & this paper\\
\hline
\end{tabular}
\vspace{.5em}
\label{tab1}
\caption{Summary of parameterized algorithms for \cvc.}
\end{table*}

\paragraph{Our results} The main result of this paper 
is a deterministic algorithm solving \cvc in $O^*(2^k)$ time.
Moreover, when we allow exponential space, in the same 
running time we can solve weighted and counting versions
of the \cvc problem, which was not possible
with the previously fastest randomized algorithm of~\cite{focs}.

\defproblemugoal{\ccvc (\ccvcshort)}{An undirected graph $G=(V,E)$, an integer $k$.}{$k$}
{ Find the number of connected vertex covers of cardinality at most $k$.} 

\defproblemugoal{\wcvc (\wcvcshort)}{An undirected graph $G=(V,E)$, a weight function $\weight:V \rightarrow \mathbb{R}_+$ and an integer $k$.}{$k$}
{ Find a minimum weight connected vertex cover of cardinality at most $k$.} 

\begin{theorem}
\label{thm:main-exp-w}
\wcvc can be solved in \break $O(2^kk(|V|+|E|))$ time and $O(2^kk)$ space.
\end{theorem}

\begin{theorem}
\label{thm:main-exp-c}
\ccvc can be solved in \break $O(2^kk(|V|+|E|))$ time and $O(2^kk)$ space.
\end{theorem}

Recently Cygan et al~\cite{seth} have shown that unless the {\em Strong
Exponential Time Hypothesis} (SETH) fails, it is not possible to count the
number of connected vertex covers of size at most $k$ in $O^*((2-\varepsilon)^k)$ time,
for any constant $\varepsilon > 0$.
Consequently our counting algorithm is tight under SETH,
which is an example of few parameterized problems with nontrivial solutions
for which there exists an evidence of optimality.

When restricted to polynomial space, we prove that the weighted variant can still be solved in $O^*(2^k)$ running time,
assuming weights are polynomially bounded integers.

\begin{theorem}
\label{thm:main-poly}
\wcvc with polynomially bounded integer weights can be solved in $O(2^k n^{O(1)})$ time and polynomial space.
\end{theorem}

\paragraph{Related work}
{\sc Vertex Cover} is one of the longest studied problem in the parameterized complexity.
The currently fastest known parameterized algorithm
for the {\sc Vertex Cover} problem is due to Chen et al.,
running in $O(1.2738^k+kn)$ time~\cite{chen-vc}.
Recently, new parameterizations of {\sc Vertex Cover} are considered,
when the parameter is $k-|M|$~\cite{almost2sat-fpt}, where $M$ is a maximum cardinality matching,
or $k-\mathrm{LP}$, where $\mathrm{LP}$ is the optimum value of a natural linear programming relaxation~\cite{vclp1,vclp2}.

A notion very close to fixed parameter tractability, or even a subfield of it, is kernelization.
We call a polynomial time preprocessing routine a kernel, if given an instance $I$
with parameter $k$ the algorithm produces a single instance $I'$ with parameter $k'$,
such that $I'$ is a YES-instance iff $I$ is a YES-instance, and moreover $|I'|+k' \le g(k)$.
It is well known that a problem admits a kernel if and only if it is kernelizable,
however we are mostly interested in kernelization algorithms with the function $g$ being a polynomial.
Unfortunately, for \cvc no polynomial kernel exists as shown by Dom et al.~\cite{colours-and-ids}, unless \compass.

\paragraph{Organization} In Section~\ref{sec:alg} we prove Theorem~\ref{thm:main-exp-w}.
For the sake of presentation we describe small differences needed to solve the counting variant,
that is to prove Theorem~\ref{thm:main-exp-c}, in separate Section~\ref{sec:counting}.
Next, in Section~\ref{sec:poly-space} we prove Theorem~\ref{thm:main-poly}
and finally, we finish the article with conclusions and open problems in Section~\ref{sec:conclusions}.

\paragraph{Notation.} We use standard graph notation.
For a graph~$G$, by~$V(G)$ and~$E(G)$ we denote its vertex and edge sets, respectively.
When it is clear which graph we are describing we use $n$ as the number of its vertices
and $m$ as the number of its edges.
For~$v \in V(G)$, its neighborhood~$N(v)$ is defined as~$N(v) = \{u: uv\in E(G)\}$, and~$N[v] = N(v) \cup \{v\}$ is the closed neighborhood of~$v$.
We extend this notation to subsets of vertices:~$N[X] = \bigcup_{v \in X} N[v]$ and~$N(X) = N[X] \setminus X$.
For a set~$X \subseteq V(G)$ by~$G[X]$ we denote the subgraph of~$G$ induced by~$X$.
For a set~$X$ of vertices or edges of~$G$, by~$G \setminus X$ we denote the graph
with the vertices or edges of~$X$ removed; in case of vertex removal, we remove
also all the incident edges.
For two subsets of vertices $X,Y \subseteq V$ by $E(X,Y)$ we denote the set of edges
with one endpoint in $X$ and the other in $Y$. In particular by $E(X,X)$ we 
denote the set of edges with both endpoints in $X$.

\section{Algorithm}
\label{sec:alg}

In this section we prove Theorem~\ref{thm:main-exp-w}.
As the starting point we use the iterative compression technique in Section~\ref{sec:alg:1}.
As a consequence we are left with a problem, where additionally each instance
is equipped with a connected vertex cover $Z$ of size at most $k+2$.
In Section~\ref{sec:alg:2} we show how to take advantage of the set $Z$
by showing a natural algorithm, solving a bipartite Steiner tree problem
as a subroutine (described in Section~\ref{sec:alg:4}).
The key part of the proof of Theorem~\ref{thm:main-exp-w}
is the time complexity analysis of the presented algorithm,
which relies on a combinatorial lemma proved in Section~\ref{sec:alg:3}.

\subsection{Iterative compression}
\label{sec:alg:1}

We start with a standard technique in
the design of parameterized algorithms, that is, 
iterative compression, introduced by Reed et al.~\cite{reed:ic}.
Iterative compression was also the first step of the 
Monte Carlo algorithm for \cvc~\cite{focs}.

We define a {\em compression problem}, where the input additionally contains a connected vertex cover $Z \subseteq V$.
The name {\em compression} might be misleading in our case, since in the problem
definition below we are not explicitly interested in compressing the solution,
but we want to find a minimum weight connected vertex cover using the size of $Z$
as our structural parameter. In particular not only we use the fact that $Z$ is a vertex cover (which ensures that $V\setminus Z$ is an independent set), but 
also we use the fact that $G[Z]$ is connected, which is crucial in the time
complexity analysis of our algorithm.

\defproblemugoal{\compwcvc (\compwcvcshort)}
{An undirected graph $G=(V,E)$, a weight function $\weight:V \rightarrow \mathbb{R}_+$, an integer $k$ and a connected vertex cover $Z \subseteq V$ of $G$.}{$|Z|$}
{Find a minimum weight connected vertex cover of cardinality at most $k$.} 

In Section~\ref{sec:alg:2} we prove the following lemma providing a parameterized algorithm for the above compression problem.

\begin{lemma}
\label{lem:compression-exp}
\compwcvcshort can be solved in $O(2^{|Z|}k(|V|+|E|))$ time and $O(2^{|Z|}k)$ space.
Moreover, when the weight function is uniform, we can solve the problem in $O(2^{|Z|}(|V|+|E|))$
time and $O(2^{|Z|})$ space.
\end{lemma}


Having the above lemma we show how to efficiently find a connected vertex cover
of size at most $k$ (if it exists).

\begin{lemma}
\label{lem:compression}
Given an undirected graph $G=(V,E)$ and an integer $k$ one can find a connected
vertex cover of size at most $k$, or verify that it does not exist,
in $O(2^kk(|V|+|E|))$ time and $O(2^k)$ space.
\end{lemma}

\begin{proof}
First, let us assume that $G$ does not contain isolated vertices, since
we can remove them.
Moreover we can assume that $G$ is connected, since if $G$ contains
at least two connected components (and no isolated vertices)
then it can not admit a connected vertex cover of any size.
Therefore, let $V=\{v_1,...,v_n\}$ be an ordering of vertices, such that for each $1 \le i \le n$,
the graph $G[V_i]$ is connected, where $V_i = \{v_i,\ldots,v_n\}$.
For $1 \le i \le n$ let $G_i$ be the graph $G$, with vertices of $V_i$ identified to a single
vertex.
Alternatively, we can say that $G_i$ comes from a contraction of the set of edges
of a spanning tree of $G[V_i]$.
Since \cvc is closed under edge contractions, we infer that if there is no
connected vertex cover of size at most $k$ in $G_i$, then
clearly there is no connected vertex cover of size at most $k$ in $G$.

We are going to construct a sequence of sets $X_i \subseteq V(G_i)$
of size at most $k$, such that $X_i$ is a connected vertex cover of $G_i$.
First, observe that the set $X_1=\emptyset$ is a connected vertex cover of $G_1$ of size at most $k$.
Next, let us consider each value of $i=2,\ldots,n$ one by one.
Observe that there is an edge $e$ in $E(G_i)$, such
that the graph $G_{i-1}$ is exactly the graph $G_i$ with the edge $e$ contracted.
In particular as $e$ we may take any edge between $v_{i-1}$ and $V_i$.
Let $x$ be the vertex in $G_{i-1}$ which corresponds to the set $V_{i-1}$
and let $y$ be the vertex in $G_i$ corresponding to the set $V_i$.
We claim that $Z=(X_{i-1}\setminus \{x\}) \cup \{v_i,y\}$ is a connected
vertex cover of $G_i$ of size at most $k+2$.
Since $|X_{i-1}| \le k$ the bound on the size of $Z$ holds.
Moreover, since $X_{i-1}$ is a vertex cover of $G_{i-1}$,
the set $Z$ is a vertex cover of $G_i$.
Finally, $G_i[Z]$ is connected, because either $x$ is contained in $X_{i-1}$,
or a neighbour of $x$ belongs to $X_{i-1}$, or $x$ is an isolated vertex 
which means that $i=2$ and then $Z=V(G_2)$ induces a connected subgraph.

If, for a fixed $i$, we use Lemma~\ref{lem:compression-exp} for the \compwcvcshort instance $(G_i,\omega,k,Z)$,
with $\omega$ being a uniform unit weight function,
then in $O(2^{|Z|}(n+m))=O(2^k(n+m))$ time and $O(2^{|Z|})=O(2^k)$ space we can find a set $X_i$,
which is a connected vertex cover of $G_i$ of cardinality at most $k$,
or verify that no connected vertex cover of cardinality at most $k$ in the graph $G$ exists.
Since $G_n=G$, the set $X_n$ is a connected vertex cover of $G$ of size at most $k$,
which we can find in $O(2^kn(n+m))$ time, because we use Lemma~\ref{lem:compression-exp}
exactly $n-1$ times.
In order to reduce the polynomial factor from $n(n+m)$ to $k(n+m)$ observe,
that if we order the set $V$, such that the set $\{v_{i-\ell+1},\ldots,v_n\}$
forms a connected vertex cover of the graph $G$, then as the set $X_{i-\ell+1}$
we can set a singleton set containing the vertex corresponding to $V_{i-\ell+1}$
and reduce the number of rounds in the inductive process from $n$ to $\ell$.
However, a simple $O(n+m)$ time $2$-approximation of the \cvc problem
is known~\cite{guha-khuller}, which just takes as the solution
the set of internal nodes of a depth first search tree of the given graph\footnote{For the sake of 
completeness in Appendix~\ref{app:cvc-apx} we present a proof of correctness of this algorithm.}.
Therefore, assuming a vertex cover of size at most $k$ exists,
we can find a connected vertex cover of size at most $2k$ in $O(n+m)$ time
and consequently reduce the number of rounds of the inductive process to at most $2k$,
which leads to $O(2^kk(n+m))$ time complexity.
\qed
\end{proof}

By Lemma~\ref{lem:compression} we can find a connected vertex cover $Z$ of size at most $k$, if it exists.
Consequently we can use Lemma~\ref{lem:compression-exp}, which proves Theorem~\ref{thm:main-exp-w}.


\subsection{Compression algorithm}
\label{sec:alg:2}

In this section we present a proof of Lemma~\ref{lem:compression-exp}.
The advantage we have while solving \compwcvcshort instead
of \wcvcshort is the additional set $Z$, which forms a connected
vertex cover of $G$ and the size of $Z$ is our new parameter.
We show how to use the set $Z$ as an insight into the structure of the graph 
and solve compression problem efficiently.
The algorithm itself is straightforward, but the crucial part
of its time complexity analysis lies in the following combinatorial
bound, which we prove in Section~\ref{sec:alg:3}.

\begin{lemma}
\label{lem:bound}
For any connected graph $G=(V,E)$ we have
\begin{equation}
\sum_{\substack{V_1 \subseteq V \\ E(G[V\setminus V_1])=\emptyset}} 2^{|\cc(G[V_1])|} \le 3\cdot 2^{|V|-1},
\end{equation}
where by $\cc(H)$ we denote the set of connected components of a graph $H$.
\end{lemma}

Observe, that in the above lemma we sum over all sets $V_1$, that form a vertex cover of $G$.
The second tool we use in the proof of Lemma~\ref{lem:compression-exp}
is the following lemma solving the node-weighted Steiner tree problem in bipartite graphs,
where both the terminals and non-terminals form independent sets. 
The proof of it can be found in Section~\ref{sec:alg:4}.

\begin{lemma}
\label{lem:bipst}
Let $G=(V,E)$ be a bipartite graph and $T \subseteq V$
be a set of terminals, such that $T$ and $V \setminus T$ are independent sets.
For a given weight function $\weight : V\setminus T -> \mathbb{R}_+$
and an integer $k$ in $O(2^{|T|}k(|V|+|E|))$ time and $O(2^{|T|}k)$ space we can find a minimum weight
subset $X \subseteq V \setminus T$ of cardinality at most $k$, such 
that $G[T \cup X]$ is connected, or verify that such a set does not exist.
Moreover for a uniform weight function $\weight$ we improve the running time to $O(2^{|T|}(|V|+|E|))$
and space usage to $O(2^{|T|})$.
\end{lemma}

Having Lemmas~\ref{lem:bound} and~\ref{lem:bipst} we can prove Lemma~\ref{lem:compression-exp}.

\begin{proof}[of Lemma~\ref{lem:compression-exp}]
Similarly as in the proof of Lemma~\ref{lem:compression} we may assume that 
the graph $G$ is connected.
We start with guessing, by trying all $2^{|Z|}$ possibilities,
a subset $Z_1$ of $Z$ that is a part of a connected vertex cover
and denote $Z_0 = Z\setminus Z_1$.

First, let us consider a special case, that is $Z_1=\emptyset$.
Then we need to take the whole set $V\setminus Z$ to cover
the edges $E(Z_1,V\setminus Z)$, since each vertex of $V \setminus Z$ has
at least one neighbour in $Z$ (otherwise it would be isolated).
It is easy to verify whether $(V\setminus Z)$ is a connected vertex cover of size at most $k$.

Therefore, we assume that $Z_1 \neq \emptyset$ and moreover $E(Z_0,Z_0)=\emptyset$,
since otherwise there is no vertex cover disjoint with $Z_0$.
Let us partition the set $V \setminus Z$ into $V_1=(V\setminus Z) \cap N(Z_0)$ and $V_0 = (V\setminus Z) \setminus V_1$.
Less formally, we split the vertices of $V\setminus Z$ depending 
on whether they have a neighbour in $Z_0$ or not.
Since we need to cover the edges adjacent to $Z_0$,
any vertex cover disjoint with $Z_0$ contains all the vertices of $V_1$.

Observe, that if there exists a vertex $v \in V_1$, such that $N(v) \subseteq Z_0$,
no vertex cover disjoint with $Z_0$ is connected, since the vertex $v$
can not be in the same connected component as any vertex of $Z_1$,
meaning that this choice of $Z_0$ is invalid (see Fig.~\ref{fig1}).
Consequently each vertex in $V_1$ has at least one neighbour in $Z_1$.
Moreover, $Z_1 \cup V_1$ forms a vertex cover of the graph $G$, 
as $V_0 \cup Z_0$ is an independent set.
Hence we want to investigate how $Z_1 \cup V_1$ can be complemented with vertices of $V_0$,
to make the vertex cover induce a connected subgraph.
Let $G'$ be the graph $G[Z_1 \cup V_0 \cup V_1]$ with connected components of $G[Z_1 \cup V_1]$
contracted to single vertices.
Denote the vertices corresponding to contracted components of $G[Z_1 \cup V_1]$
as $T$.
Note that $G'$ is bipartite, since $G[V_0]$ is an independent set.
By Lemma~\ref{lem:bipst} we can find a minimum weight set $X \subseteq V_0$
of cardinality at most $(k-|Z_1|-|V_1|)$, such that $G'[T \cup X]$ is connected,
which is equivalent to $G[Z_1 \cup V_1 \cup X]$ being connected.
Observe, that the size of the set $T$ is upper bounded by the number of connected components
of the induced subgraph $G[Z_1]$, as each vertex of $V_1$ has at least one neighbour in $Z_1$.
Therefore, by Lemma~\ref{lem:bipst}, for a fixed choice of $Z_1$ we can find
the set $X$ in 
$O(2^{|T|}(k-|Z_1|-|V_1|)(|V(G')|+|E(G')|)=O(2^{|\cc(G[Z_1])|}k(|V(G)|+|E(G)|))$ time and $O(2^{|T|}(k-|Z_1|-|V_1|))=O(2^kk)$ space.
Moreover for a uniform weight function, by Lemma~\ref{lem:bipst}, the running time is $O(2^{|\cc(G[Z_1])|}(|V(G)|+|E(G)|))$
and space usage is $O(2^k)$.

\begin{figure}[htp]
\begin{center}
\includegraphics[scale=0.8]{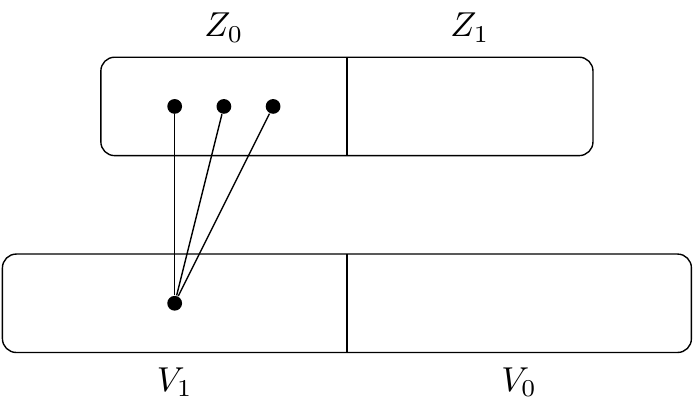}
\caption{An example of invalid choice of $Z_0$, since a vertex of $V_1$ has neighbours in $Z_1$.}
\label{fig1}
\end{center}
\end{figure}

Summing up the running time over all the choices of $Z_1$, for which $Z_0$ is an independent set,
by Lemma~\ref{lem:bound} applied to the graph $G[Z]$ we prove the total running time 
of our algorithm is $O(2^kk(|V(G)|+E|(G)|))$ for a general weight function and $O(2^k(|V(G)|+|E(G)|))$
for a uniform weight function.  \qed
\end{proof}

\subsection{Combinatorial bound}
\label{sec:alg:3}

Now we prove Lemma~\ref{lem:bound}, where we reduce the trivial $3^{|V|}$ bound to $3 \cdot 2^{|V|-1}$,
by using a similar idea, as was previously used for {\sc Bandwidth}~\cite{bw-wg,bw-tcs} and \cvc~\cite{focs}.

\newcommand{\io}{\mathbf{io}}
\newcommand{\ii}{\mathbf{ii}}
\newcommand{\oo}{\mathbf{o}}
\newcommand{\aaa}{\mathbf{a}}
\newcommand{\bbb}{\mathbf{b}}

\begin{proof}[of Lemma~\ref{lem:bound}]
Note, that we may rewrite the sum we want to bound as follows:

$$\sum_{\substack{V_1 \subseteq V \\ E(G[V\setminus V_1])=\emptyset}} 
     2^{|\cc(G[V_1])|}  = |\{(V_1,\mathcal{C}) : V_1 \subseteq V, \mathcal{C} \subseteq \cc(G[V_1]), E[G[V\setminus V_1]]=\emptyset\}|\,.$$

That is we count the number of pairs $(V_1,\mathcal{C})$, such that $V_1$
forms a vertex cover of $G$ and $\mathcal{C}$ is any subset
of connected components of the subgraph induced by $V_1$.
Denote the set of all pairs $(V_1,\mathcal{C})$ we are counting as $\mathcal{S}$.
Observe, that we can easily construct an injection $\phi$ from $\mathcal{S}$
to $\{\ii, \io, \oo\}^{|V|}$,
where for a pair $(V_1,\mathcal{C})$ as $\phi((V_1,\mathcal{C}))(v)$
we set:
\begin{itemize}
  \item $\ii$ (in-in) when $v \in V_1$ and the connected component 
  of $G[V_1]$ containing $v$ belongs to $\mathcal{C}$,
  \item $\io$ (in-out) when $v \in V_1$ and the connected component 
  of $G[V_1]$ containing $v$ does not belong to $\mathcal{C}$,
  \item $\oo$ (out) when $v \not\in V_1$.
\end{itemize}
Having any function $f:V \rightarrow \{\ii, \io, \oo\}$
we can reconstruct a pair $(V_1,\mathcal{C})$ (if it exists),
such that $\phi((V_1,\mathcal{C}))=f$.
However, the injection $\phi$ is not a surjection, for at least two reasons.
Consider any $f \in \phi(\mathcal{S})$.
Firstly, for any edge $uv \in E$, we have $f(u) \in \{\ii,\io\}$ 
or $f(v) \in \{\ii,\io\}$, since otherwise $V_1$ is not a vertex cover of $G$.
Secondly, for any edge $uv \in E$, if we have $f(u) \in \{\ii,\io\}$,
then either $f(v) = \oo$ or $f(v)=f(u)$, because
if both $u$ and $v$ belong to $V_1$, then they are are part
of exactly the same connected component $C$ of $G[V_1]$, and therefore
knowing $f(u)$ we can infer whether $C \in \mathcal{CC}$ or $C \not\in \mathcal{CC}$.

Let us formalize the intuition above, to prove that for almost each
vertex we have at most two, instead of three possibilities.
Consider a spanning tree $T$ of $G$ and root it in an arbitrary vertex $r$.
We construct the following function $\phi':\mathcal{S} \rightarrow \{\ii,\io,\oo\} \times \{\aaa,\bbb\}^{V\setminus \{r\}}$.
For a given pair $(V_1,\mathcal{C}) \in \mathcal{S}$ we set $\phi'((V_1,\mathcal{C}))=(\phi((V_1,\mathcal{C}))(r),f)$, where the function $f:V\setminus \{r\} \rightarrow \{\aaa,\bbb\}$ is defined in a top-down manner, regarding the tree $T$, as follows.
Let $v \in V\setminus \{r\}$ and denote $p \in V$ as the parent of $V$ in $T$.
\begin{itemize}
  \item If $p \in V_1$, then if $v \in V_1$, we set $f(v)=\aaa$ and otherwise (if $v \not\in V_1$), we set $f(v)=\bbb$.
  \item If $p \not\in V_1$, then we have $v \in V_1$ (since otherwise
    $V_1$ would not be a vertex cover), 
  and if the connected component of $G[V_1]$ containing $v$
  belongs to $\mathcal{CC}$, then $f(v) = \aaa$, otherwise $f(v) = \bbb$.
\end{itemize}
Since $\phi'$ is also a surjection, we have $|\mathcal{S}| \le 3\cdot 2^{|V|-1}$,
and the lemma follows. An example showing both functions $\phi,\phi'$
is depicted in Fig.~\ref{fig2}.\qed

\begin{figure}[htp]
\begin{center}
\includegraphics[scale=1]{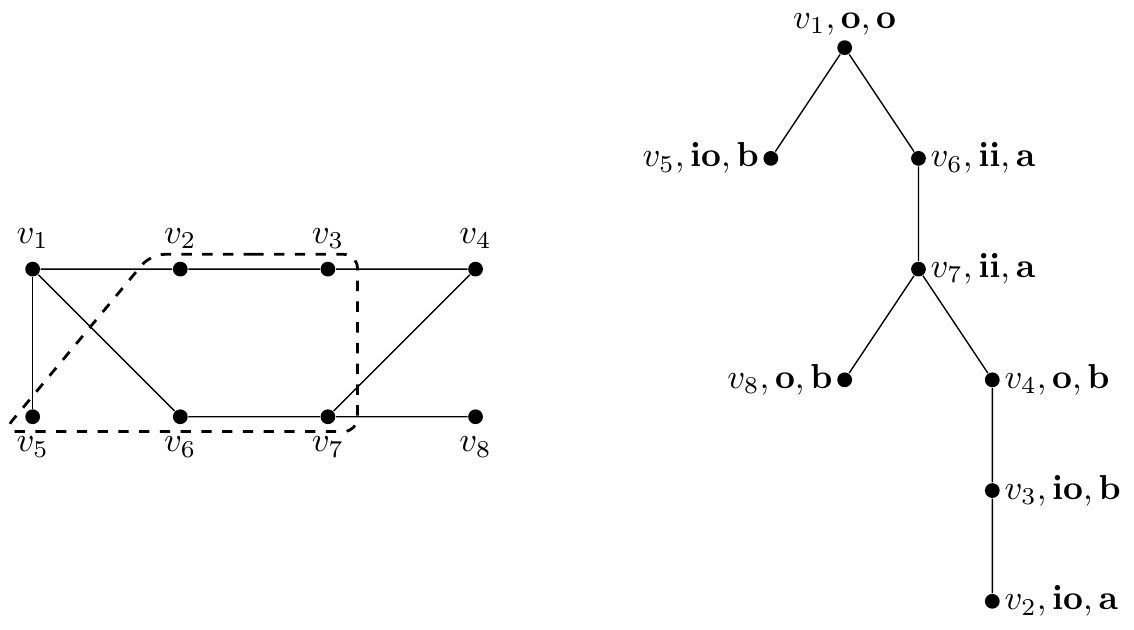}
\caption{The set $V_1$ is enclosed within the dashed border, whereas $\cc(G[V_1])=\{\{v_2,v_3\},\{v_5\},\{v_6,v_7\}\}$ and $\mathcal{C}=\{\{v_6,v_7\}\}$. On the right there is a tree $T$ rooted at $v_1$, where for each vertex
values assigned by $\phi((V_1,\mathcal{C}))$ and $\phi'((V_1,\mathcal{C}))$ are given.
Note that the for the root both $\phi((V_1,\mathcal{C}))$ and $\phi'((V_1,\mathcal{C}))$ assign exactly the same value.}
\label{fig2}
\end{center}
\end{figure}
\end{proof}

\subsection{Bipartite Steiner Tree}
\label{sec:alg:4}

Here we prove Lemma~\ref{lem:bipst}, which concerns the following bipartite variant
of the node-weighted Steiner tree problem.

\defproblemugoal{\wbipst}{An undirected bipartite graph $G=(V,E)$, a weight function $\weight:V \rightarrow \mathbb{R}_+$, an integer $k$ 
  and a set of terminals $T \subseteq V$, such that both $T$ and $V\setminus T$ are independent sets in $G$.}
{$|T|$}
{ Find a minimum weight subset $X \subseteq V\setminus T$ of size at most $k$, such that $G[T \cup X]$ is connected.}

\begin{proof}[of Lemma~\ref{lem:bipst}]
By a dynamic programming routine, for each subset $T_0 \subseteq T$ and integer $0 \le j \le k$ 
we compute the value $t(T_0,j)$, defined as the minimum weight of a subset $X \subseteq V \setminus T$, satisfying:
\begin{itemize}
  \item $|X|=j$,
  \item $N(X)=T_0$,
  \item $G[T_0 \cup X]$ is connected.
\end{itemize}
Less formally, the value $t(T_0,j)$ is the minimum weight of a set $X$ of cardinality exactly $j$,
such that $G[T_0 \cup X]$ induces a connected subgraph, and there is no edge from $X$ to $T\setminus T_0$.
Observe, that $\min_{1 \le j \le k} t(T,j)$ is the minimum weight solution for the \wbipst problem,
therefore in the rest of the proof we describe how to compute all the $(k+1)2^{|T|}$ values $t$ efficiently.

Initially, for each $t_0 \in T$ we set $t(\{t_0\},0):=0$, while all other values in the table $t$
are set to $\infty$.
Next, consider all the subsets $T_0 \subseteq T$ in the order of their increasing cardinality,
and for each integer $0 \le j < k$ and each vertex $v \in N(T_0)$ do
$$t(T_0 \cup N(v), j+1) := \min(t(T_0 \cup N(v), j+1), t(T_0,j) + \weight(v))\,.$$
Note, that the assumption $v \in N(T_0)$ ensures, that vertices $N(v) \setminus T_0$ 
get connected to the vertices of $T_0$.

With this simple dynamic programming routine we compute all the values $t(T_0,j)$
in $O(2^{|T|}k(|V(G)|+|E(G)|)$ time and $O(2^{|T|}k)$ space.
Note, that by standard methods we can reconstruct a set $X$ corresponding to the value $t(T,j)$
in the same running time.
Moreover, if the weight function is uniform, than the second dimension of our dynamic programming 
table is unnecessary, since the cardinality and weight of a set are equal.
This observation reduces both the running time and space usage by a factor of $k$.
\end{proof}

\section{Counting}
\label{sec:counting}

In this Section we present a proof of Theorem~\ref{thm:main-exp-c},
which is similar to the proof of Theorem~\ref{thm:main-exp-w}.

\begin{proof}[of Theorem~\ref{thm:main-exp-c}]
Similarly as in the proof of Theorem~\ref{thm:main-exp-c},
by using Lemma~\ref{lem:compression} in $O(2^kk(|V|+|E|))$ time we construct
a set $Z$, which is a connected vertex cover of $G$ of size at most $k$,
or verify that such a set does not exist.

Next, we proceed as in the proof of Lemma~\ref{lem:compression-exp},
however we have to justify the assumption that $G$ is a connected graph.
When $G$ contains at least two connected components containing at least two vertices each,
then there is no connected vertex cover in the graph $G$.
If there is one connected component containing at least two vertices, then
no connected vertex cover contains any of the isolated vertices, hence we can remove them.
Finally, when the graph contains only isolated vertices, then it admits
an empty connected vertex cover and $|V|$ connected vertex covers containing a single vertex only.

The rest of the proof of Lemma~\ref{lem:compression-exp} remains
unchanged and what we are left with is to show an $O(2^{|T|}k(|V|+|E|))$ running time algorithm
for the following \cbipst problem.

\defproblemugoal{\cbipst}{An undirected bipartite graph $G=(V,E)$, an integer $k$ 
  and a set of terminals $T \subseteq V$, such that both $T$ and $V\setminus T$ are independent sets in $G$.}
{$|T|$}
{ Find the number of subsets $X \subseteq V\setminus T$ of size at most $k$, such that $G[T \cup X]$ is connected.}

We do it similarly as in the proof of Lemma~\ref{lem:bipst}, that is for each $T_0 \subseteq T$
and each $0 \le j \le k$ we define the value $t(T_0,j)$, which is equal
to the number of subsets $X \subseteq V \setminus T$ of size exactly $j$,
such that $N(X) \subseteq T_0$ and $G[T_0 \cup X]$ is connected.
We leave the details of the dynamic programming routine to the reader.
\end{proof}

\section{Polynomial space}
\label{sec:poly-space}

The only place in our algorithm, where we use exponential space is when solving the \bipst
problem.
If, instead of using Lemma~\ref{lem:bipst} we use the algorithm of Nederlof~\cite{jesper-steiner},
running in $O(2^{|T|}n^{O(1)})$ time, we obtain an $O(2^k n^{O(1)})$ time and polynomial space algorithm
for the \cvc problem.
The algorithm by Nederlof solves also the weighted case, but only 
when the weights are polynomially bounded integers, which is enough to prove Theorem~\ref{thm:main-poly}.
Unfortunately, we are not aware of an algorithm which counts the number of solutions to the \bipst
problem in $2^{|T|}|V|^{O(1)}$ time and polynomial space (note that the algorithm of~\cite{jesper-steiner} counts the number of branching walks,
not the number of subsets of vertices inducing a solution).

\section{Conclusions and open problems}
\label{sec:conclusions}

In~\cite{focs} Cygan at al. we have shown a randomized $O(3^kn^{O(1)})$ algorithm for the {\sc Feedback Vertex Set} problem,
where we want to make the graph acyclic by removing at most $k$ vertices.
Is it possible to design a deterministic algorithm of the same running time?

The Cut\&Count technique presented in~\cite{focs} does not allow
neither to count the number of solution nor to solve problems with arbitrary real weights.
Nevertheless, for the \cvc problem we were able to solve both the weighted and counting
variants in the same running time. Is it possible to design $c^{{\mathrm{tw}}}n^{O(1)}$
time algorithms for counting or weighted variants of the connectivity problems
parameterized by treewidth for which the Cut\&Count technique can be applied?

Finally, we know that it is not possible to count the number 
of connected vertex covers of size at most $k$ in $O((2-\varepsilon)^kn^{O(1)})$ time, unless SETH fails.
Can we prove that we can not solve the decision version of the problem as well in such running time?

\section*{Acknowledgements}

We thank Daniel Lokshtanov, Marcin Pilipczuk and Micha\l{} Pilipczuk for helpful discussions.

\bibliographystyle{plain}
\bibliography{cvc}

\newpage
\appendix

\section{Approximation}
\label{app:cvc-apx}

In this section we show a simple algorithm, providing a $2$-approximation for the \cvc problem,
as observed by Guha and Khuller~\cite{guha-khuller}.

\begin{lemma}
Let $G$ be a connected graph and let $T$ be its depth first search tree.
The set of internal nodes of $T$ forms a connected vertex cover of $G$
of cardinality at most twice the size of a minimum connected vertex cover of $G$.
\end{lemma}

\begin{proof}
Let $X$ be the set of internal nodes of $T$.
Clearly $X$ is a connected vertex cover of $G$,
since there are not cross edges in any DFS tree.
We prove that there is matching of size at least $|X|/2$ in $G$, proving that there
is no connected vertex cover (even no vertex cover) of size smaller than $|X|/2$.

If the number of internal nodes at odd levels is at least
the number of internal nodes at even levels in $T$,
then we match each internal node on an odd level with its arbitrary child.
Otherwise we match each internal node on an even level with its arbitrary child.
In this way we show a matching of size at least $|X|/2$ and the lemma follows.
\end{proof}

\end{document}